\documentclass[journal,twoside,web]{ieeecolor}
\pdfoutput=1
\usepackage{lcsys}
\usepackage{cite}
\usepackage{amsmath,amssymb,amsfonts}
\usepackage{algorithmic}
\usepackage{algorithm}
\usepackage{graphicx}
\usepackage{textcomp}
\usepackage{tikz}
\usepackage{hyperref}

\usepackage{mathtools}
\newtheorem{problem}{Problem}
\newtheorem{assumption}{Assumption}

\newtheorem{proposition}{Proposition}
\newtheorem{theorem}{Theorem}

\newtheorem{definition}{Definition}

\newtheorem{remark}{Remark}
\usepackage{xcolor}

\newcommand{\R}{\mathbb{R}}

\def\BibTeX{{\rm B\kern-.05em{\sc i\kern-.025em b}\kern-.08em
    T\kern-.1667em\lower.7ex\hbox{E}\kern-.125emX}}
\begin{document}
\title{Koopman-inspired Implicit Backward Reachable Sets for Unknown Nonlinear Systems}
\author{Haldun Balim, \IEEEmembership{Student Member, IEEE}, Antoine Aspeel, \IEEEmembership{Member, IEEE}, Zexiang Liu, \IEEEmembership{Student Member, IEEE} and Necmiye Ozay, \IEEEmembership{Senior Member, IEEE}
\thanks{H.~B. is with ETH Zurich and he was a visiting scholar at the Univ. Michigan when this work was performed (e-mail: hbalim@ethz.ch). A.~A., Z.~L., and N.~O. are with 
the Electrical Engineering and Computer Science Department, Univ. of Michigan, Ann Arbor, MI (e-mails: antoinas,zexiang,necmiye@umich.edu). This work is funded by the ONR grant N00014-21-1-2431 (CLEVR-AI) and NSF Grant CNS 1931982.}}

\maketitle

\begin{abstract}
Koopman liftings have been successfully used to learn high dimensional linear approximations for autonomous systems for prediction purposes, or for control systems for leveraging linear control techniques to control nonlinear dynamics. In this paper, we show how learned Koopman approximations can be used for state-feedback correct-by-construction control. To this end, we introduce the Koopman over-approximation, a (possibly hybrid) lifted representation that has a simulation-like relation with the underlying dynamics. Then, we prove how successive application of controlled predecessor operation in the lifted space leads to an implicit backward reachable set for the actual dynamics. 
Finally, we demonstrate the approach on two nonlinear control examples with unknown dynamics.
\end{abstract}

\begin{IEEEkeywords}
reachability, constrained nonlinear control, data-driven control
\end{IEEEkeywords}

\section{Introduction}
\label{sec:introduction}
\IEEEPARstart{T}{he} goal of backward reachability analysis is to identify a set of states called the backward reachable set (BRS), which guarantees the existence of a control strategy to direct a system's trajectories towards a predetermined target region within a finite time. Having BRS can significantly simplify controller synthesis while ensuring safety. Particularly, it provides a state-feedback controller defined over the entire BRS as opposed to point-to-point planners. However, for general nonlinear systems, computation of maximal BRS is considered a challenging problem \cite{althoff2021}. To overcome this fundamental limitation, one approach is to compute inner-approximations which still guarantee existence of safe controllers at the expense of being conservative.

One strategy for addressing nonlinear control problems involves the application of Koopman operator theory to extend the use of well-studied linear system analysis~\cite{Otto2021}. Particularly, a lifting function is sought which transforms the coordinates to a higher dimensional space over which the nonlinear dynamics flows in a linear fashion. However, for an arbitrary nonlinear system, such a lifting function is in general infinite-dimensional. Therefore, in practice, it is only possible to construct finite-dimensional approximations, which in-turn introduces an approximation error.
By bounding the approximation error, we provide a novel Koopman-inspired approach to compute correct-by-construction inner approximations of the BRS of discrete-time nonlinear systems. We also show how this approach can be used for unknown systems where lifted local linear approximations to the dynamics are computed from data and used for BRS computation.

\subsection{Related work}

Creating linear approximations of nonlinear systems is a highly recognized subject in the field of systems and control. There are various approaches such as Taylor approximation based linearization, feedback linearization~\cite{Khalil:1173048} and linearization through state immersion~\cite{immersion_lin}. More recently, approaches motivated by Koopman operator theory start to draw revived attention in control research. Several ways to use Koopman lifted systems to address nonlinear control problems using linear methods are proposed in literature~\cite{korda2018linear, Brunton2022}. 
Koopman-like lifted systems are also used to identify forward reachable sets~\cite{bak2021reachability} and invariant sets of autonomous systems \cite{sankaranarayanan2016change,wang2023computation}.

Backward reachability analysis has also been extensively studied as a useful tool to solve constrained control problems. Existing techniques include set-based methods~\cite{blanchini2007,kurzhanskiy2011reach,althoff2021,liren_cz}, discrete-abstractions~\cite{girard2012controller,liu2016finite}, Hamilton-Jacobi (HJ) reachability~\cite{hjb, mitchell2005}, and more recent work for unknown dynamics~\cite{shafa2022reachability}. However, most of the existing methods cannot be extended to general nonlinear systems, especially when the state space dimension is large~\cite{hjb}. All in all, there is a trade-off between generality, scalability, and conservativeness, and new methods striking a different balance between these factors are needed.

\subsection{Notation}
For a matrix $A$, $A_i$ denotes its $i^\text{th}$ row. For a vector-valued function $f$, $f^i$ is the $i^\text{th}$ component of $f$. The Lipschitz constant of a function $f$ is denoted by $L_f$. A norm is represented by the notation $\|x\|$, when it is applied to a matrix it indicates induced norm. The notation $\mathcal{B}(c, r)$ denotes a closed norm-ball centered at point $c$ with a radius of $r$. The symbol $\oplus$ is used for Minkowski sum. When $\times$ is applied to sets it is to indicate cartesian product. 

\section{Problem Statement}
\label{sec:defn_ps}
We consider discrete-time systems of the form
\begin{align} \label{eqn:sys} 
\Sigma:\ x^+=f(x,u),
\end{align}
with state $x\in \mathcal{X}\subseteq \R^n$ and input $u\in \mathcal{U}\subseteq \R^m$. We use uncertain systems of the form 
\begin{align} \label{eqn:uncert_sys}
\Sigma_a:\ x^+=g(x,u,w)
\end{align}
as approximations of a given system $\Sigma$, where the disturbance input $w\in W\subseteq \R^l$ accounts for the mismatch. For an uncertain system, one-step backward reachable sets are defined as follows.

\begin{definition}[\cite{blanchini2007}]
Given an uncertain system $\Sigma_a$, a target set $X \subseteq \mathcal{X}$ and a state-input constraint set $S_{xu}= S_{x}\times S_{u} \subseteq \mathcal{X}\times \mathcal{U}$, the \emph{one-step backward reachable set} (BRS) $Pre_{\Sigma_a}(X, S_{xu})$ of the set $X$ with respect to $\Sigma_a$ and constraint set $S_{xu}$ is given by
\begin{align} \label{eqn:pre} 
\text{Pre}_{\Sigma_a}(X,S_{xu})&=\{x\in S_{x}\mid \exists u\in S_{u}: g(x,u,W) \subseteq X\}.
\end{align}
\end{definition}
We denote $\text{Pre}_{\Sigma}^0(X,S_{xu})=X$ and define recursively the \emph{$k$-step backward reachable set} ($k$-step BRS) of $X$ by
\begin{align} \label{eqn:pre_k} 
\text{Pre}_{\Sigma}^k(X,S_{xu})&=\text{Pre}_{\Sigma}(\text{Pre}_{\Sigma}^{k-1}(X, S_{xu}),S_{xu}).
\end{align}
The definition of BRS applies to systems in \eqref{eqn:sys} just by taking the trivial disturbance $W=\{0\}$. For linear systems and polytopic sets $X$ and $S_{xu}$, the $k$-step BRS of $X$ in \eqref{eqn:pre_k} is again a polytope, making it computable relatively efficiently \cite{koutsoukos_safety_2003, 1618830}. However, if the system $\Sigma$ is nonlinear, the $k$-step BRS can be highly nonconvex, which makes its computation challenging. The focus of this work is to find inner approximations of the $k$-step BRS for nonlinear systems.
\begin{problem} \label{prob:main}
Given a nonlinear system $\Sigma$, a polytopic target set $X_0$, and a polytopic constraint set $S_{xu}$, find inner approximations of the $k$-step BRS {$X_k=\text{Pre}_{\Sigma}^k(X_0,S_{xu})$}.

\end{problem}
In the next section, we recast Problem \ref{prob:main} into a problem of computing BRSs for uncertain linear systems with polytopic sets using ideas from Koopman operator theory.

\section{Koopman Over-approximations}

Inspired by Koopman operator theory, several recent works propose to approximate a nonlinear system $\Sigma$ by a higher-dimensional linear system. More specifically, for the nonlinear system $\Sigma$, there may exist a lifting function $\psi:\mathcal{X}\rightarrow \R^p$ for which the dynamics of $\psi(x)$ is approximately linear, i.e.,
$$
\psi(f(x,u)) \approx A\psi(x) + Bu.
$$

Given a lifting function $\psi:\mathcal{X}\rightarrow \R^p$, and the system matrices $A$ and $B$, we define the approximation error $E_{A,B}(x, u)$ by
\begin{align}\label{eq:error}
    E_{A,B}(x,u) = \psi(f(x,u)) - A\psi(x) - Bu.
\end{align}
When the lifting function $\psi$ and system matrices $A$ and $B$ are selected properly, the approximation error can be very small over a domain of interest (see examples in \cite{korda2018linear}). 
The evolution of the lifted state $z=\psi(x)$ can be captured by an uncertain linear system 
\begin{align} \label{eqn:sys_linear} 
    \Sigma_{lin}: z^+ = A z + Bu + w,
\end{align}
with $z\in \R^{p}$, $u\in \mathcal{U} \subseteq \R^{m}$ and $w \in W\subseteq \R^{p}$ if the set $W\subseteq \R^p$ bounds the error $E_{A,B}(x,u)$ over given state and input domains $\mathcal{X}_{sub}$ and $\mathcal{U}_{sub}$. 
\begin{definition} \label{def:lifting} 
  The tuple $(\psi,\Sigma_{lin})$ of a lifting function $\psi:\mathcal{X} \rightarrow \R^{p}$ and a linear system $\Sigma_{lin}$ is an \emph{Koopman over-approximation} of the system $\Sigma$ over subdomains $\mathcal{X}_{sub} \subseteq \mathcal{X}$ and $\mathcal{U}_{sub} \subseteq \mathcal{U}$ if for all $(x,u)\in \mathcal{X}_{sub}\times \mathcal{U}_{sub}$, we have 
  \begin{align} \label{eqn:over-approx} 
	E_{A,B}(x,u)\in  W. 
  \end{align}
\end{definition}
When the lifting function $\psi$ is clear from the context, we say that $\Sigma_{lin}$ is a Koopman over-approximation of $\Sigma$.
Here we do not require the subdomain $\mathcal{X}_{sub}$ to be forward invariant or controlled invariant. 
If $\mathcal{U}_{sub}$ and $W$ are all $\{0\} $, then $\psi(\cdot)$ is just a finite-dimensional Koopman eigenmapping for the autonomous system $x^+=f(x, 0)$ over $\mathcal{X}_{sub}$ \cite{mauroy2020koopman}. 

To compute the BRS of $\Sigma$ using the Koopman over-approximation, we need to map the target and constraint sets $X$ and $S_{x}$ from $\mathcal{X}$ to $\R^p$. A straightforward method is to find the images $\psi(X)$ (or $\psi(S_x)$) of $X$ (or $S_x$) with respect to $\psi$, which is computationally challenging since $\psi$ is typically nonlinear. Instead, we propose a more relaxed and flexible way to lift the sets in $\mathcal{X}$ to the higher-dimensional space $\R^p$.

\begin{definition} \label{def:implict_set}
    Given a lifting function $\psi:\mathcal{X} \rightarrow\R^{p}$, a set $Z \subseteq \R^{p}$ is a \emph{$\psi$-implicit inner approximation} of a set $X \subseteq \mathcal{X}$ if 
	   $\{ x \mid  \psi(x) \in Z\} \subseteq X.$
If these two sets are equal, then $Z$ is called a \emph{$\psi$-implicit representation} of $X$. \end{definition}

Note that implicit inner approximations (or representations) of a set may not be unique. The following assumption is made in the remainder of this work, which allows one to construct an implicit representation of any subset of $\mathcal{X}$ easily.
\begin{assumption} \label{asp:linear_inversion} 
The lifting function $\psi$ in Definition \ref{def:lifting} has a linear left inverse. That is, there exists a matrix $C \in \R^{n\times p}$ such that for all $x\in \mathcal{X}$, $C\psi(x)=x$. 
\end{assumption}
Assumption \ref{asp:linear_inversion} can be satisfied by including the states $x$ of $\Sigma$ as part of the outputs of $\psi(x)$. The following proposition shows how to utilize this assumption to construct implicit representations. 
\begin{proposition} \label{prop:implict} Under Assumption \ref{asp:linear_inversion}, for any subset $X \subseteq\mathcal{X}$, the set $Z = \{z \mid Cz\in X\} \subseteq \R^{p}$ is a $\psi$-implicit representation of $X$. In particular, if $X$ is a polytope, $Z$ gives a polytopic $\psi$-implicit representation of $X$.
\end{proposition}
\begin{proof}
By definition of $Z$, the set $\{x \mid \psi(x)\in Z\}$ is equal to $\{x \mid C\psi(x)\in X\}$, which is further equal to $ X$ since $C\psi(x) = x$ by Assumption \ref{asp:linear_inversion}.
\end{proof}
The next theorem shows how easily we can control the nonlinear system $\Sigma$ using a Koopman over-approximation $\Sigma_{lin}$ of $\Sigma$ and $\psi$-implicit representations of sets. 
\begin{theorem} \label{thm:control}
Suppose that $(\psi(x), \Sigma_{lin})$ is a Koopman over-approximation of $\Sigma$ over $\mathcal{X}_{sub} \times \mathcal{U}_{sub}$. Let $Z$ be a $\psi$-implicit inner approximation of a target set $X$. Then, for any state $x\in \mathcal{X}_{sub}$ of $\Sigma$, if there exists an input $u\in \mathcal{U}_{sub}$ such that $(A\psi(x)+Bu)\oplus W \subseteq Z$, then the same $u$ steers the next state $x^+=f(x,u)$ of $\Sigma$ to $X$.
\end{theorem}
\begin{proof}
Suppose that there exists $(x,u)\in \mathcal{X}_{sub} \times \mathcal{U}_{sub}$ such that $(A\psi(x)+Bu) \oplus W \subseteq Z$. By the definition of Koopman over-approximation, $\psi(f(x,u))\in (A\psi(x)+Bu) \oplus W \subseteq Z$, which further implies that $f(x,u)\in X$ since $Z$ is a $\psi$-implicit inner approximation of $X$.
\end{proof}
We call the set of inputs $u\in \mathcal{U}_{sub}$ in Proposition \ref{thm:control} such that $(A\psi(x)+Bu)\oplus W \subseteq Z$ the \emph{admissible input set} $\mathcal{A}(x,Z)$ of $x$ with respect to the implicit target set $Z$. Under Assumption \ref{asp:linear_inversion}, when $X$ and $W$ are polytopes, $\mathcal{A}(x,Z)$ can be easily computed via standard polytope operations, thanks to Proposition \ref{prop:implict}. The next theorem draws a connection between the one-step BRS for $\Sigma$ and that for $\Sigma_{lin}$.
\begin{theorem} \label{thm:BRS} 
	Let $(\psi,\Sigma_{lin})$ be a Koopman over-approximation of $\Sigma$ over $ S_{xu}=S_x\times S_u$. If $Z$ and $S_{z}$ are $\psi$-implicit inner approximations of $X$ and ${S}_x$, respectively, the one-step BRS $\text{Pre}_{\Sigma_{lin}}(Z, S_{zu})$, with $S_{zu}=S_{z}\times {S}_u$, is a $\psi$-implicit inner approximation of the one-step BRS $\text{Pre}_{\Sigma}(X, S_{xu})$. That is,
\begin{align}\label{eqn:implicit_x}
   \{x \mid \psi(x) \in \text{Pre}_{\Sigma_{lin}}(Z, S_{zu})\} \subseteq \text{Pre}_{\Sigma}(X, S_{xu}).  
\end{align}
\end{theorem}
\begin{proof}
Let $x_0$ be such that $\psi(x_0)\in\text{Pre}_{\Sigma_{lin}}(Z, S_{zu})$. Then (i) $\psi(x_0)\in S_z$ which implies $x_0\in {S}_x$; and (ii) there exists $u_0\in {S}_u$ such that $(A\psi(x_0)+Bu_0)\oplus W\subseteq Z$. Because $(\psi,\Sigma_{lin})$ is a Koopman over-approximation of $\Sigma$, then $\psi(f(x_0,u_0))\in (A\psi(x_0)+Bu_0)\oplus W$ which is a subset of $Z$. This shows that $x_0\in\text{Pre}_{\Sigma}(X, S_{xu})$. Since $x_0$ is arbitrary in $\{x \mid \psi(x) \in \text{Pre}_{\Sigma_{lin}}(Z, S_{zu})\}$, the proof is complete.
\end{proof}
By construction, every state $x$ in the set on the LHS of \eqref{eqn:implicit_x} has a nonempty admissible input set $\mathcal{A}(x,Z)$, which can be easily extracted under Assumption \ref{asp:linear_inversion}. 
Moreover, according to Proposition \ref{prop:implict}, the polytopic sets $X$ and $S_x$ in Problem \ref{prob:main} have polytopic implicit representations
$$
Z_0=\{Cx \mid x\in X\},\ S_z=\{Cx \mid x\in S_x\}.
$$Then, by applying Theorem \ref{thm:BRS} recursively, the $k$-step BRS $Z_k = \text{Pre}_{\Sigma_{lin}}^k(Z, S_{zu})$ of $Z$ provides a polytopic inner approximation of the $k$-step BRS $X_k=\text{Pre}_{\Sigma}(X, S_{xu})$. Their relationship is illustrated in Fig. \ref{fig:main}. 

So far, we assume that there exists a single Koopman over-approximation $\Sigma_{lin}$ over $S_{xu}$. 
Similar to local linearizations in gain-scheduled control, one can also find a set of local Koopman over-approximations $\{\Sigma_{lin,i}\}_{i=1}^{N_k}$ (with the same lifting function $\psi$) at \emph{each} step, where each $\Sigma_{lin,i}$ over-approximates $\Sigma$ over a subdomain $S_{x,i}\times S_{u,i}$ of $S_{xu}$. Then, Theorem \ref{thm:BRS} can be easily extended to show that if $Z_{k-1}$ is an $\psi$-implicit inner approximation of $X_{k-1}$, the union $Z_{k}$ of the one-step BRSs $\text{Pre}_{\Sigma_{lin,i}}(Z_{k-1}, S_{zu,i})$\footnote{Here $S_{zu,i}=S_{z,i}\times S_{u,i}$ with $S_{z,i}$ a $\psi$-implicit representation of $S_{x,i}$.} for all $i$ is a $\psi$-implicit inner approximation of the $k$-step BRS $X_k$ of $\Sigma$. That is,
\begin{align*}
\{x \mid \psi(x)\in \cup_{i=1}^{N_k} \text{Pre}_{\Sigma_{lin,i}}(Z_{k-1}, S_{zu,i})\}\subseteq X_{k}.
\end{align*} 
 
 Later, our numerical examples show that using local Koopman over-approximations allows us to compute much larger BRSs, while as a cost, if $N_k>1$, the $k$-step implicit BRS $Z_k$ needs to be represented by a union of polytopes (instead of a single polytope).

\section{Computational Approach for Finding Koopman Over-approximations}
In this section, we first discuss how to obtain global Koopman over-approximations. When the system dynamics is unknown, we discuss how one may use data to obtain such over-approximations. Then, we provide a method to efficiently find local Koopman over-approximations based on the obtained global Koopman over-approximations.

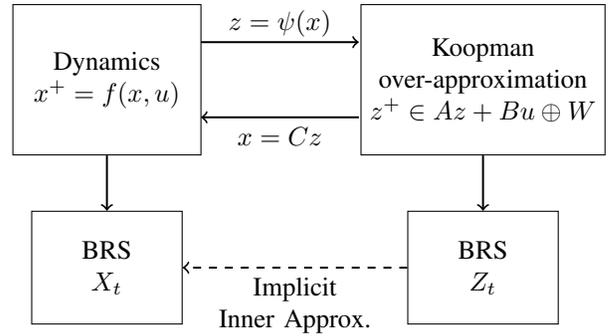
\begin{figure}[thb]
\begin{center}
\begin{tikzpicture}[node distance=3cm]
    \node [rectangle, draw, minimum width=2.5cm, minimum height=2cm, align=center] (rect1) at (0,0) {Dynamics \\ $x^+ = f(x,u)$};
    \node [rectangle, draw, minimum width=2.5cm, minimum height=2cm, align=center] (rect2) at (5,0) {Koopman \\over-approximation \\ $z^+ \in Az+Bu\oplus W$};
    \node [rectangle, draw, minimum width=2cm, minimum height=1.5cm, align=center] (rect3) at (0,-2.5) {BRS\\ $X_t$};
    \node [rectangle, draw, minimum width=2cm, minimum height=1.5cm, align=center] (rect4) at (5,-2.5) {BRS\\ $Z_t$};

    \node[] at (2.3, 0.75) {$z=\psi(x)$};
    \node[] at (2.3, -0.75) {$x = Cz$};
    \node[align=center] at (2.5, -3) {Implicit \\ Inner Approx.};
    
    \draw [->, line width=0.75] (1.25, 0.5) to (3.35, 0.5);
    \draw [->, line width=0.75] (3.35, -0.5) to (1.25, -0.5);
    \draw [->, line width=0.75] (rect1) -- (rect3);
    \draw [->, line width=0.75] (rect2) -- (rect4);

    \draw [dashed, ->,line width=0.75] (rect4) to (rect3);
    
    \end{tikzpicture}
\end{center}
\caption{Relation between different systems and their BRSs.}
\label{fig:main}
\end{figure}

\subsection{Computing Koopman over-approximations}
\label{subsec:koop_rep}

Trivially, if we can bound the error in \eqref{eq:error} in the entire constraint set, i.e., for all $(x,u)\in S_{xu}$, $E_{A,B}(x,u)\in \mathcal{B}(c,\epsilon)$, then the linear system $\Sigma_{lin}$ defined by $A$, $B$ and $W=\mathcal{B}(c,\epsilon)$ is a Koopman over-approximation of $\Sigma$ over $S_{xu}$. Note that such a bounding ball always exists if $f$ and $\psi$ are continuous and $S_{xu}$ is compact (see Theorem 4.15 in \cite{rudin1976principles}). Moreover, it can be computed analytically or using optimization techniques (e.g., \cite{boyd2004convex}) when the dynamics of $\Sigma$ is known.

Let us now assume the dynamics is unknown and we are given a finite data set $\mathcal{D} = \{({x}_i,{u}_i,{x}_i^+)\}_{i=1}^N$ where $(x_i,u_i)$ is sampled in $S_{xu}$ and ${x}_i^+=f({x}_i,{u}_i)$. We denote by $\mathcal{D}_{xu}$, the state-input pairs in the data set $\mathcal{D}$. Let us further assume that we are given a valid Lipschitz constant $L_{E_{A,B}}$ for the error function $E_{A,B}$. We use the dispersion of the data set together with the Lipschitz constant of the error to evaluate the error bound.

\begin{definition}[\cite{lavalle2006planning}] \label{defn:dispersion}
Given the constraint set $S_{xu}$ and data set $\mathcal{D}_{xu}$, the \emph{dispersion} $b$ of $\mathcal{D}_{xu}$ in $S_{xu}$ is defined by $$b=\sup_{(x,u)\in S_{xu}} \min_{(\bar{x},\bar{u})\in\mathcal{D}_{xu}} \| (x,u) - (\bar{x},\bar{u})\|.$$
\end{definition}
Note that we do not require $\mathcal{D}_{xu}\subseteq S_{xu}$.
The next theorem translates the error on the data set to an error bound.

\begin{theorem} \label{thm:global_oaka}
Given a constraint set $S_{xu}$, a data set $\mathcal{D}$ for which $\mathcal{D}_{xu}$ has dispersion $b$ in $S_{xu}$, a lifting function $\psi$, and matrices $A$, $B$ of compatible dimensions, and consider the Lipschitz constant $L_{E_{A,B}}$. Define
$$e(A,B)=\min_c \max_{(\bar{x},\bar{u},\bar{x}^+)\in\mathcal{D}}\|\psi(\bar{x}^+) - A\psi(\bar{x})-B\bar{u}-c\|,$$
with $c^*$ as the minimizer.
Let $\epsilon=L_{E_{A,B}}b+e(A,B)$. Then the linear system $\Sigma_{lin}$ defined by $A$, $B$, and $W=\mathcal{B}(c^*,\epsilon)$ is a Koopman over-approximation of $\Sigma$ over the domain $S_{xu}$.
\end{theorem}
\begin{proof} Following \cite{knuth2021planning}, for any $(\bar{x}, \bar{u})\in\mathcal{D}_{xu}$ and $(x,u)\in S_{xu}$, one can write
\begin{align*}
&\|E_{A,B}(x,u) -c^*\| \\
 &\leq \|E_{A,B}(x,u) - E_{A,B}(\bar{x},\bar{u}) \| + \| E_{A,B}(\bar{x},\bar{u})-c^*\| \\
 &\leq L_{E_{A,B}} \|(x,u)-(\bar{x},\bar{u})\| + \| \psi(\bar{x}^+) - A\psi(\bar{x})-B\bar{u}-c^*\|,
\end{align*}
where the last inequality follows from the definition of Lipschitz constant and the definition of $\bar{x}^+$. Finally, by definition of $b,e$ and $\epsilon$, we conclude that $\| E_{A,B}(x,u)-c^*\|\leq \epsilon$, for all $(x,u)\in S_{xu}$. 
\end{proof}

Given a data set $\mathcal{D}$, Theorem \ref{thm:global_oaka} provides a way to find a Koopman over-approximation over a domain $S_{xu}$. First, solve the following convex optimization problem to obtain the system parameters:
\begin{align} \label{eq:global}
    A,B=\arg\min_{A,B} \; e(A,B).
\end{align}
Second, find (an over-estimation of) the Lipschitz constant $L_{E_{A,B}}$ (see Remark \ref{remark:EVT} below). Third, Theorem \ref{thm:global_oaka} provides a $W$ which leads to a Koopman over-approximation.

\begin{remark}\label{remark:EVT}
While it is not possible to obtain deterministic Lipschitz constant bounds from finite amount of data, estimation of Lipschitz constants is an active research area \cite{wood1996estimation, weng2018evaluating}. For instance, given $A$, $B$, $\psi$ and a data set $\mathcal{D}$, if $\mathcal{D}_{xu}$ is independent and identically distributed, statistical methods such as scenario approach or extreme value theory can be used to obtain an estimate \cite{knuth2021planning, 9691930}. The same statistical techniques can also be applied to directly estimate a bound on $E_{A,B}$ from data. 
In the next subsection, we show how a given valid estimate $L_{E_{A,B}}$ can be used to locally calibrate $A$, $B$ in an efficient way.
\end{remark}

\subsection{Computing local Koopman over-approximations}
\label{subsec:local_koop_rep}

The error bound for globally estimated $A$, $B$ matrices can be large, possibly resulting in conservative BRS inner approximations.
For this reason, we want to establish a smaller error bound by adapting the globally estimated $\Sigma_{lin}$ parameters to smaller localities. We rely on the fact that for matrices $\tilde{A}$ and $\tilde{B}$, the function $\Delta_{\tilde{A},\tilde{B}}$ defined by
\begin{align*}
\Delta_{\tilde{A},\tilde{B}}(x,u) &= E_{\tilde{A},\tilde{B}}(x,u) -E_{A,B}(x,u)\\
&= (A-\tilde{A})\psi(x) + (B-\tilde{B})u,
\end{align*}
is independent of the nonlinear dynamics $f$.
The following theorem is useful for developing a method to compute a local Koopman over-approximation over a subset of $S_{xu}$.

\begin{theorem} \label{thm:local_oaka} 
Given a subdomain $\tilde{S}\subseteq S_{xu}$, a data set $\mathcal{D}$ with dispersion $b$ in $S_{xu}$, a lifting function $\psi$, and matrices $A$, $B$, $\tilde{A}$, $\tilde{B}$ of compatible dimensions, consider the Lipschitz constants $L_{\Delta_{\tilde{A},\tilde{B}}}$ and $L_{E_{A,B}}$. Define $\tilde{\mathcal{D}}\subseteq\mathcal{D}$ such that $\mathcal{\tilde{D}}_{xu}=\mathcal{D}_{xu}\cap\left(\tilde{S}\oplus \mathcal{B}(0,b)\right)$, and
$$
\tilde{e}(\tilde{A},\tilde{B})=\min_{\tilde{c}} \max_{(\bar{x},\bar{u},\bar{x}^+)\in\tilde{\mathcal{D}}}
\|\psi(\bar{x}^+) - \tilde{A}\psi(\bar{x})-\tilde{B}\bar{u}-\tilde{c}\|,
$$
with $\tilde{c}^*$ as the minimizer. Let $\tilde{\epsilon}=( L_{\Delta_{\tilde{A},\tilde{B}}} + L_{E_{A,B}}) b+ \tilde{e}(\tilde{A},\tilde{B})$. Then the linear system $\tilde{\Sigma}_{lin}$ defined by $\tilde{A}$, $\tilde{B}$, and $\tilde{W}=\mathcal{B}(\tilde{c}^*,\tilde{\epsilon})$ is a Koopman over-approximation of $\Sigma$ over the subdomain $\tilde{S}$.
\end{theorem}

\begin{proof}
First, let us show that $\tilde{\mathcal{D}}_{xu}$ has dispersion at most $b$ in $\tilde{S}$. If $b$ is finite, then $S$ is bounded and so is $\tilde{S}$. Let $(x,u)$ be in the closure of $\tilde{S}$ and let $(\bar{x},\bar{u})$ be its closest neighbor in $\mathcal{D}_{xu}$. By definition of $b$, $\| (x,u) - (\bar{x},\bar{u}) \|\leq b$. But then, $(\bar{x},\bar{u})\in \tilde{S}\oplus\mathcal{B}(0,b)$ and finally $(\bar{x},\bar{u})\in\tilde{\mathcal{D}}_{xu}$, which proves the claim.

Then, for any $(x,u)\in \tilde{S}$ and $(\bar{x},\bar{u})\in \tilde{S}\cap\mathcal{D}$, one can write
\begin{align*}
&\| E_{\tilde{A},\tilde{B}}(x,u) - \tilde{c}^* \| \\
&= \| E_{\tilde{A},\tilde{B}}(x,u) - E_{\tilde{A},\tilde{B}}(\bar{x},\bar{u}) + E_{\tilde{A},\tilde{B}}(\bar{x},\bar{u}) - \tilde{c}^*\| \\
&\leq \| E_{\tilde{A},\tilde{B}}(x,u) - E_{\tilde{A},\tilde{B}}(\bar{x},\bar{u}) \| + \| E_{\tilde{A},\tilde{B}}(\bar{x},\bar{u}) - \tilde{c}^* \| \\
&= \| E_{\tilde{A},\tilde{B}}(x,u) - E_{A,B}(x,u) +  E_{A,B}(\bar{x},\bar{u}) \\
&\indent - E_{\tilde{A},\tilde{B}}(\bar{x},\bar{u}) + E_{A,B}(x,u) - E_{A,B}(\bar{x},\bar{u}) \| \\
&\indent+ \| E_{\tilde{A},\tilde{B}}(\bar{x},\bar{u}) - \tilde{c}^* \|\\
&\leq \|\Delta_{\tilde{A},\tilde{B}}(x,u) - \Delta_{\tilde{A},\tilde{B}}(\bar{x},\bar{u}) \| \\
&\indent+ \| E_{A,B}(x,u) - E_{A,B}(\bar{x},\bar{u}) \| + \|E_{\tilde{A},\tilde{B}}(\bar{x},\bar{u}) - \tilde{c}^* \| \\
&\leq ( L_{\Delta_{\tilde{A},\tilde{B}}} + L_{E_{A,B}} ) \|(x,u)-(\bar{x},\bar{u})\| \\
&\indent+ \| E_{\tilde{A},\tilde{B}}(\bar{x},\bar{u}) - \tilde{c}^* \|.
\end{align*}
Since $\tilde{\mathcal{D}}_{xu}$ has dispersion at most $b$ in $\tilde{S}$, then for all $(x,u)\in \tilde{S}$, $\| E_{\tilde{A},\tilde{B}}(x,u) -c^*\| \leq \tilde{\epsilon}$, which concludes the proof.
\end{proof}

Theorem \ref{thm:local_oaka} can be used to find $\tilde{A}$, $\tilde{B}$ and $\tilde{c}^*$ over a subdomain $\tilde{S}$ by solving the following optimization problem

\begin{align} 
\tilde{A},\tilde{B},\tilde{c}^*=\arg & \min_{\tilde{A},\tilde{B},\tilde{c}} \Big\{ \left( L_{\Delta_{\tilde{A},\tilde{B}}} + L_{E_{A,B}} \right) b \notag \\
&+ \max_{(\bar{x},\bar{u},\bar{x}^+)\in\tilde{\mathcal{D}}}
\|\psi(\bar{x}^+) - \tilde{A}\psi(\bar{x})-\tilde{B}\bar{u}-\tilde{c}\| \Big\}. \label{eq:local_true}
\end{align}

In general the dependence on $L_{\Delta_{\tilde{A},\tilde{B}}}$ makes the optimization problem \eqref{eq:local_true} hard to solve. However, the following theorem states that if $S_{xu}$ is equipped with the infinity norm and if for each dimension $i=1,\dots,p$ a Lipschitz constant $L_{E^i_{A_i,B_i}}$ of $E^i_{A_i,B_i}(x,u)=\psi^i(f(x,u))-A_i\psi(x)-B_iu$ is known, then a local Koopman over-approximation can be found by solving $p$ linear programs.

\begin{theorem}\label{thm:local_oaka:inf}
Given a subdomain $\tilde{S}\subseteq S_{xu}$ equipped with the infinity norm, a data set $\mathcal{D}$ with dispersion $b$ in $S_{xu}$, a lifting function $\psi$, and matrices $A$, $B$, $\tilde{A}$, $\tilde{B}$ of compatible dimensions, consider the Lipschitz constants $L_\psi$ and $L_{E^i_{A_i,B_i}}$ for $i=1,\dots,p$. Define $\tilde{\mathcal{D}}\subseteq\mathcal{D}$ such that $\tilde{\mathcal{D}}_{x,u}=\mathcal{D}_{xu}\cap\left( \tilde{S}\oplus\mathcal{B}(0,b) \right)$, and for $i=1,\dots,p$, define
$$
\tilde{e}_i(\tilde{A}_i,\tilde{B}_i)=\min_{\tilde{c}_i} \max_{(\bar{x},\bar{u},\bar{x}^+)\in\tilde{\mathcal{D}}}
|\psi^i(\bar{x}^+) - \tilde{A}_i\psi(\bar{x})-\tilde{B}_i\bar{u}-\tilde{c}_i|,
$$
with $\tilde{c}^*_i$ as the minimizer. Let
\begin{align*}
\bar{\epsilon}_i &= \left( \|A_i^\top-\tilde{A}_i^\top\|_1 L_\psi + \|B_i^\top-\tilde{B}_i^\top\|_1 + L_{E^i_{A_i,B_i}} \right)b \\
&\indent + \tilde{e}_i(\tilde{A}_i,\tilde{B}_i).
\end{align*}
Then the linear system $\bar{\Sigma}_{lin}$ defined by $\tilde{A}$, $\tilde{B}$, $\bar{W}=\tilde{c}^*\oplus\bigtimes_{i=1}^p[-\bar{\epsilon}_i,\bar{\epsilon}_i]$ is a Koopman over-approximations of $\Sigma$ over the subdomain $\tilde{S}$.
\end{theorem}

\begin{proof}
Following the same reasoning as the proof of Theorem \ref{thm:local_oaka}, for all $(x,u)\in \tilde{S}$,
\begin{equation}\label{eq:proof:local:E}
|E^i_{\tilde{A},\tilde{B}}(x,u) - c^*_i | \leq ( L_{\Delta^i_{\tilde{A}_i,\tilde{B}_i}} + L_{E^i_{A_i,B_i}}) b + \tilde{e}_i(\tilde{A}_i,\tilde{B}_i).
\end{equation}
We are going to show that
\begin{equation}\label{eq:proof:local:L}
L_{\Delta^i_{\tilde{A}_i,\tilde{B}_i}}\leq \|A_i^\top-\tilde{A}_i^\top\|_1 L_\psi + \|B_i-\tilde{B}_i\|_1.
\end{equation}
Indeed, for all $(x^1,u^1),(x^2,u^2)\in \tilde{S}$, using the triangular inequality, the definition of Lipschitz constant and the definition of infinity norms, one can write
\begin{align*}
&|\Delta^i_{\tilde{A}_i,\tilde{B}_i}(x^1,u^1)-\Delta^i_{\tilde{A}_i,\tilde{B}_i}(x^2,u^2)| \\
&=|(A_i-\tilde{A}_i)\left(\psi(x^1)-\psi(x^2)\right) + (B_i-\tilde{B}_i)(u^1-u^2)| \\
&\leq |(A_i-\tilde{A}_i)\left(\psi(x^1)-\psi(x^2)\right)| + |(B_i-\tilde{B}_i)(u^1-u^2)| \\
&\leq \|A_i^\top-\tilde{A}_i^\top\|_1 L_\psi\|x^1-x^2\|_\infty + \|B_i^\top-\tilde{B}_i^\top\|_1\|u^1-u^2\|_\infty \\
&\leq \left( \|A_i^\top-\tilde{A}_i^\top\|_1 L_\psi + \|B_i^\top-\tilde{B}_i^\top\|_1 \right) \| (x^1-x^2, u^1-u^2) \|_\infty,
\end{align*}
which proves relation \eqref{eq:proof:local:L}. Combining \eqref{eq:proof:local:E} and \eqref{eq:proof:local:L} with the definition of $\bar{\epsilon}_i$, we have $|E^i_{\tilde{A},\tilde{B}}(x,u) - \tilde{c}^*_i | \leq \bar{\epsilon}_i$ for all $(x,u)\in \tilde{S}$, which concludes the proof.
\end{proof}

An important consequence of Theorem \ref{thm:local_oaka:inf} is that a Koopman over-approximation can be found by solving the following linear programs: For $i=1,\dots,p$,
\begin{align}
\tilde{A}_i,\tilde{B}_i, \tilde{c}^*_i=&\arg\min_{\tilde{A}_i,\tilde{B}_i,\tilde{c}_i} \Big\{ \big( \|A_i^\top-\tilde{A}_i^\top\|_1 L_\psi + \|B_i^\top-\tilde{B}_i^\top\|_1 + \notag \\
&L_{E^i_{A_i,B_i}} \big) b \notag \\
&+\max_{(\bar{x},\bar{u},\bar{x}^+)\in\tilde{\mathcal{D}}} |\psi^i(\bar{x}^+) - \tilde{A}_i\psi(\bar{x})-\tilde{B}_i\bar{u}-\tilde{c}_i| 
\Big\}. \label{eq:local_oaka:AB}
\end{align}

\section{Results and Discussion}
\label{sec:results}

In this section, our method is illustrated by computing the BRS of two nonlinear dynamical systems, the forced Duffing oscillator and the inverted pendulum, and compared with the HJB method on these examples. Our code that generates the figures and implements our algorithm is available at: \url{https://github.com/haldunbalim/KoopmanBRS}. The reported computation times are obtained with a laptop with a Quad-Core Intel i7 CPU and 16 GB of RAM.

\subsection{Forced Duffing oscillator}
We consider forced Duffing oscillator \cite{Guckenheimer2002} with dynamics
\begin{equation*}
    \mathbf{\dot x} = \begin{bmatrix}y\\ 2x -2x^3 -0.5y + u\end{bmatrix},
\end{equation*}
discretized using Runge-Kutta 4(5) scheme with a $0.025$ time step. The state is $\mathbf{x} = \begin{bmatrix}x, y\end{bmatrix}^\top \in [-0.5, 0.5] \times [-1.5, 1.5]$ and the input is $u \in \begin{bmatrix}-5, 5\end{bmatrix}$. The target set is $X_T = [-0.1, 0.1]\times[-0.1, 0.1]$ and the considered lifting function is $\psi(\mathbf{x}) = [x, y, x^3]^\top$. 

A Koopman over-approximation is estimated from 1000 random tuples $(\mathbf{x},u,\mathbf{x}^+)$ according to \eqref{eq:global}. Then the error $E_{A,B}$ is bounded component-wise via the extreme value theory using 200 samples (see Remark \ref{remark:EVT}). We validated the fit of an extreme value distribution with significance 0.05 using the Kolmogorov-Smirnov (KS) goodness-of-fit test~\cite{DeGroot2002} with 50 samples. The estimated error bound is $\epsilon=[0.0005, 0.0004, 0.0133]^\top$.

We use our method to compute a 10-step BRS. The results are presented in Figure~\ref{fig:duffing-brs} with a comparison with BRS obtained by the Hamilton-Jacobi (HJB) method~\cite{hjb}. Note that unlike our method, HJB requires knowledge of the dynamics and can be seen as a ground truth modulo numerical inaccuracies in PDE solutions. The HJB method took $74.41$ s, while our method ran in $0.41$ s given the Koopman over-approximation.

\begin{figure}[t]
\vspace{-0.35cm}
\centerline{\includegraphics[width=0.8\columnwidth]{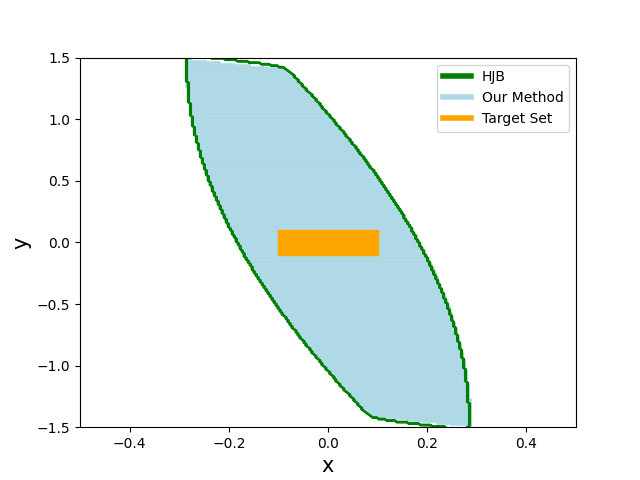}}
\caption{Backward reachable sets for the Duffing oscillator.}
\label{fig:duffing-brs}\vspace{-0.4cm}
\end{figure}

\vspace{-0.2cm}
\subsection{Inverted pendulum}

\begin{figure}[t] 
\centerline{\includegraphics[width=0.8\columnwidth]{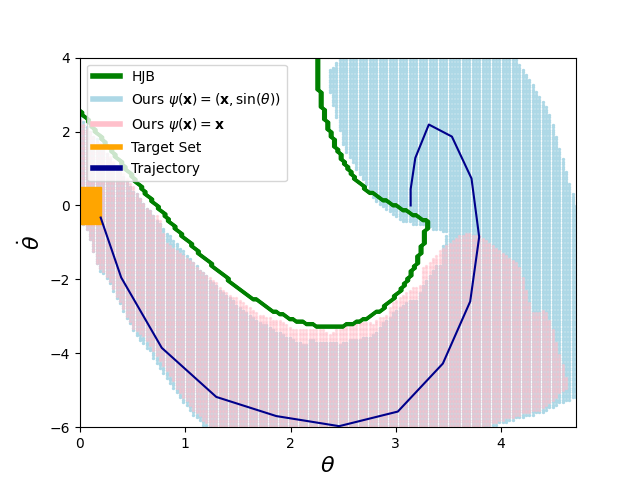}}
\caption{Backward reachable sets for the inverted pendulum.}
\label{fig:pendulum}\vspace{-0.35cm}
\end{figure}

Consider the dynamics of an inverted pendulum:
\begin{equation*}
    \mathbf{\dot x} = \begin{bmatrix} \dot \theta\\ \frac{3g}{2l}\sin(\theta) +\frac{3}{ml^2}u\end{bmatrix}.
\end{equation*}
It is discretized with time step 0.1 using explicit Euler scheme. The state is $\mathbf{x} = [\theta, \dot{\theta}]^\top \in S_x=[0,\frac{3\pi}{2}]\times[-6, 4]$, and the input is $u\in S_u=[-0.35, 0.35]$. The parameters $(m, l, g)$ are set to $(0.1, 1, 10)$ respectively. The target set is $[0, 0.2] \times [-0.5, 0.5]$ and the lifting function considered is $\psi(\mathbf{x}) = [\theta, \dot{\theta}, \sin(\theta)]^\top$.

Matrices $A$, $B$ are computed using equation \eqref{eq:global} from a data set of approximately $300,000$ samples created via grid-sampling with $0.04$ unit separation in the state dimensions and $0.08$ unit separation in the action dimension.
The component-wise Lipschitz constant $L_{E^i_{A,B}}$ of the error function is estimated using extreme value theory \cite{wood1996estimation} with $40,000$ samples for fitting and $10,000$ samples for KS test with a significance value of 0.05. We obtain $L_{E_{A,B}} = [0, 0, 0.737]^\top$.

For each local Koopman over-approximation, matrix $B$ is kept constant, while matrix $A$ is adapted using Theorem \ref{thm:local_oaka}. Since we are only changing $A$ matrix, $L_{\Delta_{\tilde{A},\tilde{B}}}$ is replaced with $L_{\Delta_{\tilde{A}}}$ which only depends on states. Therefore, we only need to multiply this term by dispersion $b_x$ computed over states.

To choose the subdomains for a given target set, we filter tuples in our dataset for which the next state falls into the target set. Then, the subdomain is chosen as a rotated bounding box around the states of these tuples using \cite{pca_box}. If the 1-norm of the radii of $\bar W$ for a subdomain is greater than $0.18$, the corresponding target set is split into two halves, through its Chebyshev center along the axis its bounding box has the greatest span. This is done to create smaller subdomains where local Koopman over-approximations found by equation \eqref{eq:local_oaka:AB} yield smaller errors. 

In Figure~\ref{fig:pendulum}, we present a comparison of the BRS computed using our method, a variant where $\psi(x)$ is the identity function, and the BRS computed using the HJB method. All three methods are used to compute BRS for a horizon of $15$ and the computation times for these approaches were $245.4$ s, $952.6$ s, and $85.4$ s, respectively. Our results indicate that the inclusion of a lifting dimension led to a significant increase in computational speed (thanks to a smaller number of subdomains being sufficient for achieving small error) and generated sets with larger volume. Finally, we show a trajectory starting from $\mathbf{x} = [\pi, 0]^\top$ that is steered to the target set using control inputs extracted according to Theorem~\ref{thm:control}. The pendulum initially moves counter-clockwise to accumulate energy before proceeding towards the target set. This intricate maneuver demonstrates the effectiveness of our methodology in tackling the problem at hand.

\section{Conclusion and Future Directions}
\label{sec:conclusion}

Inspired by Koopman operator theory, in this paper, we have introduced Koopman over-approximations for discrete-time nonlinear systems. These over-approximations allow us to use linear system backward reachability tools to compute implicit BRSs of nonlinear systems. Crucially, control inputs steering the system to the target set can be easily extracted from these implicit BRSs at run-time. We have also presented computational approaches to construct Koopman over-approximations from data when the underlying nonlinear system is unknown. Finally, we have discussed a local version of these over-approximations, which, in a sense, generalizes the hybridization approaches, such as \cite{girard2011synthesis}, in correct-by-construction control literature using lifting.

While our theoretical results hold for any given Lipschitz lifting function, a limitation of our work is the lack of a method for choosing the lifting function, which affects the size of the computed BRS. Our future work will aim to overcome this limitation by investigating how to obtain lifting functions incrementally in a way to yield monotonically better performance. From an algorithmic point of view, we plan to replace the polytopic reachability tools used in this paper to compute the BRS of the Koopman over-approximation with more efficient zonotopic ones \cite{liren_cz} to improve scalability. Finally, it is also interesting to generalize the results to handle uncertainty in dynamics and measurements.

\bibliographystyle{IEEEtran}
\bibliography{bibliography}

\end{document}